\documentclass[submission,copyright,creativecommons]{eptcs}
\providecommand{\event}{LSFA 2011}  

\usepackage{amssymb}
\usepackage{amsmath,amsthm}
\usepackage{stmaryrd}
\usepackage{enumerate}

\providecommand{\urlalt}[2]{\href{#1}{#2}} 
\providecommand{\doi}[1]{doi:\urlalt{http://dx.doi.org/#1}{#1}} 

\newcommand{\bc}{\begin{center}}
\newcommand{\ec}{\end{center}}
\newcommand{\be}{\begin{enumerate}}
\newcommand{\ee}{\end{enumerate}}
\newcommand{\bi}{\begin{itemize}}
\newcommand{\ei}{\end{itemize}}
\newcommand{\beq}{\begin{equation}}
\newcommand{\eeq}{\end{equation}}
\newcommand{\beqs}{\begin{equation*}}
\newcommand{\eeqs}{\end{equation*}}
\newcommand{\ba}{\begin{array}}
\newcommand{\ea}{\end{array}}

\newtheorem{theorem}{Theorem}
\newtheorem{proposition}{Proposition}
\newtheorem{defn}{Definition}

\newtheorem{lem}{Lemma}
\newtheorem{ejem}{Example}

\newcommand{\pt}[1]{\langle #1 \rangle}
\newcommand{\ol}[1]{\overline{#1}}

\newcommand{\C}{\sat}
\newcommand{\sat}{\mathsf{SAT}}

\newcommand{\cl}[1]{\mathsf{cl}(#1)}

\newcommand{\id}{\mathsf{Id}}
\newcommand{\efe}{\mathsf{F}}

\renewcommand{\L}{\mathcal{L}}
\newcommand{\Pe}{\mathcal{P}}

\newcommand{\R}{\mathcal{R}}
\newcommand{\I}{\mathcal{I}}
\newcommand{\F}{\mathcal{F}}
\newcommand{\K}{\mathcal{K}}

\newcommand{\M}{\mathcal{M}}
\newcommand{\Nc}{\mathcal{N}}
\newcommand{\E}{\mathcal{E}}

\newcommand{\tA}{\mathit{A}}

\newcommand{\G}{\Gamma}

\newcommand{\Eb}{\mathbb{E}}
\newcommand{\N}{\mathbb{N}}
\newcommand{\St}{\mathbb{S}}

\newcommand{\Lp}{\mathop{\sf fst}}
\newcommand{\Rp}{\mathop{\sf snd}}
\newcommand{\inl}{\mathop{\sf inl}}
\newcommand{\inr}{\mathop{\sf inr}}
\newcommand{\inm}{\mathop{\mathsf{in}}}
\newcommand{\out}{\mathop{\mathsf{out}}}

\newcommand{\pack}{\mathop{\mathsf{pack}}}
\newcommand{\open}{\mathsf{open}}

\newcommand{\pred}{\mathop{\mathsf{pred}}}
\newcommand{\fac}{\mathop{{\sf fac}}}

\newcommand{\head}{\mathop{\mathsf{head}}}
\newcommand{\tail}{\mathop{\mathsf{tail}}}
\newcommand{\from}{\mathop{\mathsf{from}}}
\newcommand{\cons}{\mathop{\mathsf{cons}}}

\newcommand{\map}{\mathsf{map}}
\newcommand{\maph}{\mathsf{maphd}}

\newcommand{\fa}{\forall}
\newcommand{\ex}{\exists}
\newcommand{\lb}{\lambda}
\newcommand{\Lb}{\Lambda} 
\newcommand{\imp}{\rightarrow}
\newcommand{\Imp}{\Rightarrow}

\newcommand{\Iff}{\Leftrightarrow}

\newcommand{\inc}{\subseteq}
\newcommand{\vacio}{\varnothing}

\newcommand{\sn}{\mathsf{SN}}

\newcommand{\case}{\mathsf{case}}

\newcommand{\eqdef}{=_{def}} 
\newcommand{\lfp}{\mathop{\mathsf{lfp}}}
\newcommand{\gfp}{\mathop{\mathsf{gfp}}}

\newcommand{\afd}{\mathsf{AF2}}
\newcommand{\afdmn}{\mathsf{AF2}^{M\mu\nu}}
\newcommand{\itm}{\mathop{\mathsf{MIt}}}
\newcommand{\coitm}{\mathop{\mathsf{MCoIt}}}
\newcommand{\recm}{\mathop{\mathsf{MRec}}}
\newcommand{\corecm}{\mathop{\mathsf{MCoRec}}}

\newcommand{\pif}{\Phi}
\newcommand{\El}[1]{E\big[#1\big]}
\newcommand{\ist}{\mathop{{\sf ist}}}
\newcommand{\sist}{\mathop{{\sf prt}}}
\newcommand{\pg}{\Phi}

\newcommand{\thii}{\Theta_I}
\newcommand{\thie}{\Theta_E}

\newcommand{\thec}{\Theta^\subseteq_E}
\newcommand{\ism}{\I_\sigma}
\newcommand{\suc}{\mathop{\mathsf{suc}}}

\newcommand{\psup}{\Phi^{\sqsupseteq}}
\newcommand{\psub}{\Phi^{\sqsubseteq}}
\newcommand{\mq}{\sqsubseteq}
\newcommand{\qm}{\sqsupseteq}

\newcommand{\cf}{\mathsf{c}}
\newcommand{\df}{\mathsf{d}}
\newcommand{\fsf}{\mathsf{f}}

\newcommand{\Jk}{\mathfrak{J}}

\newcommand{\lf}{\mathop{\mathsf{lf}}}
\newcommand{\rg}{\mathop{\mathsf{rg}}}
\newcommand{\unp}{\langle\!\!\!\!\langle\,\rangle\!\!\!\!\rangle}
\newcommand{\suma}{\mathop{{\sf sum}}}

\newcommand{\Mf}{\mathfrak{M}}

\title{Mendler-style Iso-(Co)inductive predicates: a strongly normalizing approach
\footnote{This research is being supported by PAPIIT-UNAM projects IN117711
  and IN108810.}}

\author{Favio Ezequiel Miranda-Perea \and Lourdes del Carmen Gonz\'alez-Huesca
\institute{ Departamento de Matem\'aticas, Facultad de Ciencias UNAM
\\ Circuito Exterior S/N, Cd. Universitaria, 04510, M\'exico D.F., M\'exico}
\\ \email{favio@ciencias.unam.mx \quad\quad\quad\quad luglzhuesca@ciencias.unam.mx}
}
\def\titlerunning{Mendler-style Iso-(Co)inductive predicates: a strongly normalizing approach}
\def\authorrunning{Miranda-Perea and Gonz\'alez-Huesca}
\begin{document}
\maketitle
\begin{abstract}
We present an extension of the second-order logic $\afd$ with iso-style inductive and
coinductive definitions specifically designed
to extract programs from proofs \`a la Krivine-Parigot by means of primitive
(co)recursion principles. Our logic includes primitive constructors of least and greatest
fixed points of predicate transformers, but contrary to the common approach, we do
not restrict ourselves to positive operators to ensure monotonicity, instead
we use the
Mendler-style, motivated here by the concept of monotonization of an
arbitrary operator on a complete lattice. We prove an adequacy
theorem with respect to a realizability semantics based on $\sat$ (saturated) sets and $\sat$-valued
functions and as a consequence we obtain the strong normalization property for
the proof-term reduction, an important feature which is absent in previous related work. \\ 
{\bf Keywords:} Mendler-style, (co)inductive definitions, primitive (co)recursion, strong
normalization, saturated set, monotonization, second-order logic, programming
with proofs.
\end{abstract}

\section{Introduction}

The system $\afd$ for second-order intuitionistic logic introduced by Leivant and Krivine \cite{le83,kri93}, is one 
of the most fruitful systems obtained by the Curry-Howard correspondence. It
types exactly the same terms as the system $\efe$ of Girard and Reynolds and shares with it the properties of
strong normalization and subject reduction. Its main improvement with respect to system $\efe$ 
is that it allows the extraction of programs via the programming-with-proofs
paradigm of Krivine and Parigot. This method, originally developed in
\cite{kp90} (see also \cite{le83}) ensures the correctness of programs
($\lambda$-terms) extracted from proofs of termination statements of functions
involving formal data types, that is, from proofs of totality. Well known
results ensure the extraction of programs for all functions whose termination
is provable in second order Peano arithmetic. Nevertheless this result, satisfactory
from the extensional point of view does not suffice for an intensional view
concerning programs. In $\afd$ we can get programs for all needed functions,
but these do not have necessarily the intended behavior, see \cite{par89}. To
solve this problem some extensions of $\afd$ with least fixed points ({\sf
  TTR} \cite{par92}) and  also with greatest fixed points ($\afd^{\mu\nu}$
\cite{raf93}) have been introduced. These features allow for the (co)inductive
definition of predicates and are suitable for programming with proofs. However the strong normalization is lost
due to the use of a fixed-point combinator in the proof-term system, which
encodes derivations with lambda terms. The situation is that an iterative
function $f$ can be defined within $\afd$ and therefore its extracted program
$\bar{f}$ is automatically terminating, but the extracted program for a primitive recursive function 
employs a fixed-point combinator in the extensions of $\afd$ and
therefore its termination is not obvious at all. This has lead to sophisticated methods to verify that these programs
indeed terminate \cite{ms94}, even when they fit into a well-known terminating recursion
pattern captured in G\"odel's T for the case of natural numbers and
generalized to all (co)inductive types in \cite{mp09,mat98,men87}, for example.  
The main contribution of this paper is the introduction of a new
extension of $\afd$ with primitive (co)recursion over least and greatest
fixed points, called $\afdmn$, that enjoys the strong normalization property. Instead
of using a fixed-point combinator we use the Mendler-style approach of
\cite{men87} but with two important differences: we use a natural
deduction approach, and we do not restrict ourselves to positive
operators. This shows that such syntactical restriction is irrelevant
to the strong normalization proof of the whole Mendler-system, a feature first discovered by
Matthes (\cite{mat98}, p.83) for the inductive fragment. Another contribution
of our work is the use of the iso-style, meaning that a (co)inductive
predicate and its folding/unfolding are not considered equal but
isomorphic. It is important to mention that previous extensions of $\afd$ with
(co)inductive definitions deal only with equi-style predicates, but in our opinion
the use of the iso-style is closer to the usual mechanisms of data type definition in functional
programming languages. 
As a consequence of our definition of saturated sets, the proof of the adequacy
theorem of our logic does not employ ordinal recursion. Moreover, the rules of our logic are specifically designed
to derive statements of totality of functions involving (co)inductive
predicates, that is, formulas of the form $\fa x.\Pe(x)\to\R(f(x))$.
The paper is organized as follows: in section \ref{sc:fp} we review
the required concepts of fixed-point theory needed to motivate the
definition of our logic, which is given in section \ref{sc:log}, together
with some examples of its expressivity. Section \ref{sec:sat} develops the constructions on
saturated sets employed in section \ref{sec:sem} to define an intuitionistic semantics of
the logic. Finally, we discuss related work in section \ref{sec:rw} and provide some closing remarks in section \ref{sec:concl}.

\section{Fixed-point theory}
\label{sc:fp}

In this section we recall some tools of fixed-point theory involving a
complete lattice $\pt{\L,\mq,\bigsqcap}$, where $\bigsqcap$ is the infimum
operator.  Given a monotone operator
$\Phi:\L\to\L$ the Knaster-Tarski theorem guarantees the existence of the
least (greatest) fixed-point of $\Phi$, denoted  $\lfp(\Phi)$ or $\gfp(\Phi)$, respectively. 

\begin{proposition}[Conventional (co)induction principles]\label{cor:prind}
Let $\Phi:\L\to\L$ be a monotone operator on a complete lattice
$\pt{\L,\mq,\bigsqcap}$. The following holds for every $M\in\L$.
\begin{itemize}
\item Induction: if $\Phi(M)\sqsubseteq M$ then $\lfp(\Phi)\sqsubseteq M$.
\item Extended induction: if $\Phi\big(\lfp(\Phi)\bigsqcap M\big)\sqsubseteq M$ then
  $\lfp(\Phi)\sqsubseteq M$.
\item Coinduction: if $M\sqsubseteq\Phi(M)$ then $M\sqsubseteq\gfp(\Phi)$.
\item Extended coinduction$\;$\footnote{Recall that in a complete lattice the
    supremum operator $\bigsqcup$ can be defined from the infimum operator
    $\bigsqcap$.}: 
if $M\sqsubseteq\Phi\big(\gfp(\Phi)\bigsqcup M\big)$ then
  $M\sqsubseteq\gfp(\Phi)$.
\end{itemize}
\end{proposition}

\begin{proof} Straightforward.
\end{proof}

The following concepts of monotonization of an arbitrary operator are taken from
\cite{mat98}.
\begin{defn}\label{def:mtz} Given an arbitrary operator $\Phi:\L\to\L$, we define its upper
  monotonization $\psup:\L\to\L$ and its lower monotonization $\psub:\L\to\L$
  as $\;\psup(M) = \bigsqcup \{\Phi(X)\;|\;X\sqsubseteq M\}\;$ and $\;\psub(M) = \bigsqcap \{\Phi(X)\;|\;M\sqsubseteq X\}$.
\end{defn}
\vspace{5pt}
The properties and relationships between $\Phi$ and its monotonizations are given in the following
\begin{proposition}\label{psubpsup:prop}  
If $\Phi:\L\to\L$ is an arbitrary operator then $\psub$ and $\psup$ are
monotone. Moreover,
\bi
\item For any $M\in\L$, $\;\psub(M)\sqsubseteq \Phi(M)\sqsubseteq \psup(M)$.
\item If $\Phi$ is monotone then $\;\psub=\Phi=\psup$ and if $\psub=\Phi\;$ or $\;\Phi=\psup$ then $\Phi$ is monotone. 
\item $\Phi (\lfp(\psup))\sqsubseteq \lfp(\psup)\;$ and $\; \gfp(\psub)\sqsubseteq \Phi(\gfp(\psub))$.
\ei  
\end{proposition}

\begin{proof} Straightforward.
\end{proof}

Next we justify the Mendler-style (co)induction principles by means of the
monotonizations. This justification is not present in the original
work of Mendler (\cite{men87}). However, the inductive part is discussed in \cite{mat98}. 

\begin{proposition}[Mendler (Co)induction principles]\label{pr:mprin}
$\!\!\!$The following holds for any $\Phi:\L\!\!\to\!\!\L$ and $M\!\!\in\!\L$.
\bi
\item Induction: 
if  $\;\fa X\big(X\mq M\imp \Phi(X)\mq M\big)$ then $\lfp(\psup)\mq M$.
\item Extended Induction: 
if  $\;\fa X\big(X\mq\lfp(\psup)\to X\mq M\imp \Phi(X)\mq M\big)$ then $\lfp(\psup)\mq M$.
\item Coinduction:
if $\;\fa
  X\big(M\mq X\imp M\mq \Phi(X)\big)$ then $\;M\mq\gfp(\psub)$.
\item Extended Coinduction:
if $\;\fa X\big(\gfp(\psub)\mq X\imp M\mq X\imp M\mq \Phi(X)\big)$ then $\;M\mq\gfp(\psub)$.
\ei
\end{proposition}

\begin{proof} The conventional (co)induction principles for $\psub$ and $\psup$
  yield the required principles. For details see \cite{mg11}.
\end{proof}

\section{The Logic $\afdmn$}\label{sc:log}
We present now the logic $\afdmn$, which is an extension of $\afd$ with Mendler-style
(co)inductive definitions.

\bi
\item {\em Terms}: the object terms are defined as usual from a signature
  $\Sigma$ including  function symbols $f$ of a given arity. 
\beqs
t::= x\;|\;f(t_1,\ldots,t_n)
\eeqs
\item {\em Predicates}: 
apart from the usual predicates (second-order variables or predicate symbols of a
signature $\Sigma$) 
we have comprehension predicates, inductive predicate $\mu(\Phi)$ and coinductive
predicates $\nu(\Phi)$.  
\beqs
\Pe ::= X\;|\;P\;|\;\F\;|\;\mu(\Phi)\;|\;\nu(\Phi)
\eeqs
here $\F$ is a comprehension predicate of the form $\F\eqdef\lb\vec{x}.A$, 
where $A$ is a formula and its arity is the length of the vector of variables
$\vec{x}$, this predicate intends to represent the set
$\{\vec{t}\;|\;A[\vec{x}:=\vec{t}\,]\}$. On the other hand, $\Phi$ is an arbitrary predicate
transformer, which is a {\em closed} expression of the form $\Phi\eqdef\lb X.\Pe$, 
depending on a second-order variable $X$. Observe that we do not
require any syntactic restriction, like positivity, on the occurrences of $X$
in $\Pe$.
\item {\em Formulas}: these are defined as usual
\beqs
A,B ::= \Pe(t_1,\ldots,t_n)\;|\;A\to B\;|\;
\fa xA\;|\;\fa XA
\eeqs
\item {\em On equations}: term equations are formulas which play an important
  role in the logic and are defined as usual in second-order logic: the equation
  $r=s$ stands for the formula $\fa X.X(r)\imp X(s)$. 
\ei

The judgments of the logic are of the form $\G\vdash_\Eb t:A$ where 
$\G=\{x_1:A_1,\ldots,x_n:A_n\}$ is a context of formulas
annotated by {\em proof-term} variables, $\Eb=\{r_1=s_1,\ldots,r_n=s_n\}$
is a context of equations, $A$ is a formula and $t$ is a {\em proof-term}, 
which is a lambda term not to be confused with an object term, for even when we use the same
meta-variables for both, object and proof-terms, we consider them to be two
completely separated syntactic categories.  The derivation relation is
inductively defined by means of the following inference rules, where
$A[x:=r]$ ($A[X:=\Pe]$) always denotes capture-avoiding substitution of
first-order (second-order) variables by a term (predicate) in the formula $A$. 
\bi
\item {\em Rules of $\afd$}:
\begin{equation*}
\frac{}{\G,x:A\vdash x:A\;\;(Var)}\hspace{1.3cm}
\frac{\Gamma,x:A\vdash r:B}{\Gamma\vdash \lambda xr:A\rightarrow B}\;\;(\rightarrow\! I) \hspace{1.3cm}
\frac{\Gamma\vdash r:A\rightarrow B\;\;\;\Gamma\vdash s:A}{\Gamma\vdash rs:B}\;\;(\rightarrow E)
\eeqs
\beqs
\frac{\Gamma\vdash t:A\;\;\;\;x\notin FV(\G)}{\Gamma\vdash t:\forall xA}\;(\forall I)\hspace{1.3cm}
\frac{\Gamma\vdash t:\forall xA}{\Gamma\vdash t:A[x:=r]}\;(\forall E)
\eeqs
\beqs
\frac{\Gamma\vdash t:A\;\;\;\;X\notin FV(\G)}{\Gamma\vdash t:\forall XA}\;(\forall^2 I)\hspace{1.3cm}
\frac{\Gamma\vdash t:\forall XA}{\Gamma\vdash t:A[X:=\Pe]}\;(\forall^2 E)
\eeqs

\beqs
\frac{\G\vdash_\Eb t:A[x:=r]\;\;\;\;\;\Eb\rhd r=s}{\G\vdash_\Eb
  t:A[x:=s]}\;(Eq)
\eeqs

Here $\Eb\rhd r=s$ means a derivation of $r=s$ from the set of equations $\Eb$
according to the following rules: 
\begin{itemize}
\item $\Eb\rhd r=s$, if $r=s$ is a particular case of an equation in $\Eb$.
  That is an equation of the form $r_1[\vec{x}:=\vec{t}\,]=r_2[\vec{x}:=\vec{t}\,]$ or
  $r_2[\vec{x}:=\vec{t}\,]=r_1[\vec{x}:=\vec{t}\,]$, where $r_1=r_2\in\Eb$ and
  $\vec{t}$ are arbitrary terms.
\item $r=s$ was obtained from $\Eb$ by reflexivity, transitivity or
  compatibility with functions, that is, by one of the following rules:
\beqs
\frac{}{\Eb\rhd r=r}\hspace{1.3cm}\frac{\Eb\rhd r=s\;\;\;\;\;\Eb\rhd s=t}{\Eb\rhd
  r=t}\hspace{1.3cm}\frac{\Eb\rhd r_1=s_1\;\;\ldots\;\;\Eb\rhd
  r_n=s_n}{\Eb\rhd f(r_1,\ldots,r_n)=f(s_1,\ldots,s_n)}
\eeqs
\end{itemize}
\item {\em Rules involving (co)inductive definitions}: these rules are specifically
  designed to construct (destruct) elements of an inductive (coinductive)
  predicate and to prove statements of totality of functions. Given two 
  $n$-ary\footnote{We are mostly interested in predicates
    for data types, which means $n=1$. However we present the system for any
    arity for the sake of generality.} predicates $\Pe,\R$, and a vector
  $\vec{g}$ of $n$ function symbols, the following
  notation will be used: $\Pe\inc_{\vec{g}} \R$ 
  is the formula $\fa \vec{x}.\Pe(\vec{x}\,)\to\R(\vec{g}(\vec{x}\,))$, where, in
  general, a vector application of $\vec{f}=_{def}f_1,\ldots,f_n$ to
$\vec{t}=_{def}t_1,\ldots,t_n$, denoted $\vec{f}(\vec{t}\,)$, is defined as $\vec{f}(\vec{t}\,)=_{def}f_1(t_1),\ldots,f_n(t_n)$.
In particular $\Pe\inc\R$ is the formula $\fa
  \vec{x}.\Pe(\vec{x})\to\R(\vec{x})$ or even $\Pe\to\R$, if the predicates have arity $0$. Given a predicate transformer
  $\Phi\eqdef \lb X.\Pe$ and a predicate $\R$, the application of $\Phi$ to $\R$, is defined
  by $\Phi(\R)\eqdef\Pe[X:=\R]$, clearly $\Phi(\R)$ is a
  predicate. \\The following rules are motivated by the last part of proposition
  \ref{psubpsup:prop} and  by proposition \ref{pr:mprin}, for lattices of sets. It is important to observe
  that in each rule we employ $\mu(\Phi)$ or $\nu(\Phi)$ instead of the
  expected $\mu(\Phi^\qm)$  or $\nu(\Phi^\mq)$. This choice will be justified by the semantics.
\bi
\item {\em Inductive construction and coinductive destruction}: for any
  (co)inductive predicate $\mu(\Phi)$ or $\nu(\Phi)$ of arity $n$, we assume a
  fixed set of $n$ function symbols $\vec{c}$ or $\vec{d}$, called the constructors of
  $\mu(\Phi)$ or the destructors of $\nu(\Phi)$. 
\beqs
\frac{\G\vdash r:\Phi(\mu(\Phi))(\,\vec{t}\,)}{\G\vdash\inm r:\mu(\Phi)(\vec{c}(\vec{t}\,)\,)}\;(\mu I)\hspace{1.5cm}
\frac{\G\vdash r:\nu(\Phi)(\vec{t}\,)}{\G\vdash\out r:\Phi(\nu(\Phi))(\vec{d}(\vec{t}\,))}\;(\nu E)
\eeqs
These rules correspond to the last part of proposition
\ref{psubpsup:prop}, but observe that our (co)inductive predicates are
in iso-style, due to the presence of the constructors $\vec{c}$ (destructors 
$\vec{d}$). Moreover, the equi-style can be easily
recovered by using as constructors/destructors the identity
function symbol $id$ while adding $id(x)=x$ to the equational axioms.
\item {\em Primitive recursion}: this rule is modelled after the Mendler
  extended induction
  principle given by proposition \ref{pr:mprin}.
  Here we regard a composition $f\circ c$ as a new function symbol defined by the
  equation $(f\circ c)(x)=f(c(x))$ and a composition of
  tuples $\vec{f}\circ\vec{c}$ as the tuple $f_1\circ
  c_1,\ldots,f_n\circ c_n$.
\beqs
\!\!\!\!\!\!\!\!\frac{\G\vdash s:\fa X\big(X\inc\mu(\Phi)\imp X\inc_{\vec{f}} \K\imp 
                    \Phi(X)\inc_{\vec{f}\circ \vec{c}} \K\big)\;\;\;\;\;\G\vdash
                    r:\mu(\Phi)(\vec{t}\,)}{\G\vdash \recm s\;r:\K(\vec{f}(\vec{t}\,))}\;(\mu E)
\eeqs
\item {\em Primitive corecursion}: the Mendler extended coinduction principle of
  proposition \ref{pr:mprin} inspires the following rule.
 Observe that in both rules (recursion and
  corecursion), we can recover the corresponding exact
  principle of proposition \ref{pr:mprin} by using the equi-style and by regarding $f$ as the identity
  function via the equation $f(x)=x$. 
\beqs
\!\!\!\!\!\frac{\G\vdash s:\fa X\big(\nu(\Phi)\inc X\imp\K\inc_{\vec{f}} X\imp 
                    \K\inc_{\vec{d}\circ \vec{f}} \Phi(X) \big)\;\;\;\;\;\G\vdash
                    r:\K(\vec{t}\,)}{\G\vdash \corecm s\;r:\nu(\Phi)(\vec{f}(\vec{t}\,))}\;(\nu I)
\eeqs
\ei
\item {\em Operational semantics}: To end the definition of our logic, we define the operational semantics of the
proof-term reduction, which is given by the one-step reduction relation $t\to_\beta t'$ defined as the closure of the
  following axioms under all term formers.
\beqs
\ba{rll}
(\lb xr)s & \mapsto_\beta & r[x:=s] \\
\recm s(\inm t)  & \mapsto_\beta &  s(\lb xx)(\recm s)t \\
\out (\corecm s\; t) & \mapsto_\beta & s(\lb xx)(\corecm s)t
\ea
\eeqs
Here and troughout the paper $\recm s$ means $\lb x.\recm s\,x$ and the same
is true for $\corecm s$. 

\item {\em Derived rules}: To simplify the presentation of examples we will
  employ the usual second-order encodings for conjunctions, disjunctions and
  existential formulas, which allow to obtain the following derived rules for judgements and
  operational semantics:
\beqs
\hspace{-.45cm}
\frac{\G\vdash r:A\;\;\;\G\vdash s:B}{\G\vdash\pt{r,s}:A\land B}\;\;(\land I)\hspace{.8cm}
\frac{\G\vdash s:A\land B}{\G\vdash \Lp s:A}\;(\land E_L)\hspace{.8cm}
\frac{\G\vdash s:A\land B}{\G\vdash \Rp s:B}\;(\land E_R)
\eeqs
\beqs
\frac{\G\vdash r:A}{\G\vdash\inl
  r:A\lor B}\;(\lor I_L)\hspace{.35cm}
\frac{\G\vdash r:B}{\G\vdash\inr
  r:A\lor B}\;(\lor I_R)
\eeqs

\beqs
\frac{\G\vdash r:A\lor B\;\;\;\;\G,x:A\vdash
  s:C\;\;\;\;\G,y:B\vdash t:C}{\G\vdash \mathsf{case}(r,x.s,y.t):C}\;\;(\lor E)
\eeqs
\beqs
\hspace{2cm}\displaystyle\frac{\G\vdash t:A[x:=r]}{\G\vdash\pack t:\ex x.A} \hspace{2cm}
\displaystyle\frac{\G\vdash t:\ex x.A\;\;\;\G,u:A\vdash r:B\;\;\;x\notin FV(\G,B)}{\G\vdash \open(t,u.r):B}
\eeqs
\beqs
\ba{rllcrll}
\Lp\pt{r,s} &\mapsto_\beta & r &\;\;\;\;\;\; &\Rp\pt{r,s}& \mapsto_\beta& s \\ 
\case(\inl r,x.s,y.t)& \mapsto_\beta& s[x:=r]&\;\;\;\;\;\; & \case(\inr
r,x.s,y.t)&\mapsto_\beta &t[y:=r] 
\ea
\eeqs
$$\open(\pack t,u.r)\mapsto_\beta\; \;r[u:=t] \hspace{1cm}$$
\ei

The proof-reduction behaves well with respect to the derivation relation, as
ensured by the following

\begin{proposition}[Subject-reduction of $\afdmn$]
  If $\G\vdash_\Eb t:A$ and
  $t\to^\star t'$ then $\G\vdash_\Eb t':A$.
\end{proposition}

\begin{proof}
The proof is not trivial since $\afdmn$ is formulated in Curry-style and it is
analogous to the one developed in \cite{mp09} for a similar system. 
\end{proof}

\subsection{On (Co)Iteration}
In fixed-point theory, (co)iteration can be easily derived from primitive (co)re\-cur\-sion. 
This is not the case for conventional (co)induction principles in type theory
like the ones developed in \cite{mp09} (see section 4.5 of
\cite{mat98} for a deep discussion on this subject) and therefore (co)iterators must be defined apart
from (co)recursors. For the Mendler-style, (co)iterators correspond to
the (co)induction principles of proposition \ref{pr:mprin}, and are again
superfluous (as noticed also in \cite{mat98}). Let us define $\itm
s\, r  \eqdef  \recm s'\, r$ and $\coitm s\,r\eqdef\corecm s'\,r$, where $s'\eqdef
\lb\_.s$ and $\_$ is a dummy variable. The following rules for
inference and proof-reduction are derivable:
\bi
\item Iteration 
\beqs
\frac{\G\vdash s:\fa X\big(X\inc_{\vec{f}} \K\imp 
                    \Phi\,X\inc_{\vec{f}\circ \vec{c}} \K\big)\;\;\;\;\;\G\vdash
                    r:\mu(\Phi)(\vec{t}\,)}{\G\vdash \itm s\;r:\K(\vec{f}(\vec{t}\,))}\;(\mu E^-)
\eeqs
\item Coiteration
\beqs
\frac{\G\vdash s:\fa X\big(\K\inc_{\vec{f}} X\imp 
                    \K\inc_{\vec{d}\circ \vec{f}} \Phi\,X \big)\;\;\;\;\;\G\vdash
                    r:\K(\vec{t}\,)}{\G\vdash \coitm
                    s\;r:\nu(\Phi)(\vec{f}(\vec{t}\,))}\;(\nu I^-)
\eeqs
\ei

\beqs
\itm s(\inm t)\to s(\itm s)\,t \hspace{2cm}
\out (\coitm s\; t) \to s(\coitm s)\,t
\eeqs
We will use both the (co)iteration and the primitive (co)recursion rules in the examples that we discuss next.

\subsection{Examples}
In this section we develop
some examples of (co)inductive predicates that show the expressivity of our
logic. Due to lack of space a deep discussion about the advantages and
disadvantages of both the iso-style and the equi-style is missing. Instead, we
provide some examples that show some of such (dis)advantages. 
Every program ($\lb$-term) $\ol{f}$ presented here is extracted from a
proof of totality for a function $f$ involving (co)inductive predicates and specified by a set of equations in
the logic. Moreover, the reader can verify that in each case $\ol{f}$ is operationally correct.

\begin{ejem} [Iso-inductive ad-hoc Natural Numbers]
Let $\unp\eqdef \lb x.x=\star$ where $\star$ is a fixed constant, this
comprehension predicate is called unit predicate and represents a type with unique inhabitant $\star$.
We define the predicate of natural numbers as $\N\eqdef\mu(\Phi)$ where
$\Phi\eqdef\lb X.\lb x.\unp(x)\lor X(x)$, taking the successor function $\suc$ as
constructor and $\suc(\star)=0$ as equational axiom. 
Defining $\ol{0}\eqdef \inm(\inl())$,\footnote{Sometimes an equation
  is involved directly in a judgment and we agree to give it the void
  proof-term $()$ as code.} and $\ol{\suc}\eqdef \lb x.\inm(\inr x)$ 
we can show that $\vdash \ol{0}:\N(0)$ and $\vdash \ol{\suc}:\fa
x.\N(x)\to \N(\suc x)$. We call this an ad-hoc definition, for zero is in the image of the
successor and therefore our representation is not compatible with Peano's
axioms. This is an unpleasant feature which can be avoided at some cost (see
example \ref{ej:catnats}). However, operationally, our definition is
adequate. For instance, the sum and factorial are programmed as follows:
\bi
\item Sum: from $\Eb_{\suma}=\{\suma n\,0=n,\;\suma n\,(\suc m)=\suc\,(\suma n\,m)\}$ , we get 
$\vdash_{\Eb_{\suma}} \ol{\suma}:\fa n.\fa x.\N(n)\to\N(x)\to\N(\suma
n\,x)$, where $\ol{\suma}\eqdef\lb n.\itm s$ and $s\eqdef\lb y\lb
  z.\case(z,u.n,v.\ol{\suc}(yv))$.  This program behaves correctly: $\ol{\suma}\;n\,\ol{0}\to^\star n\;$
  and $\;\ol{\suma}\;n\,(\ol{\suc}\,m)\to^\star \ol{\suc}(\ol{\suma}\,n\,m)$.
\item Factorial: using the equations 
$\Eb_{\mathsf{fac}}=\{\mathsf{fac}\;0 = 1,\;\mathsf{fac}\;(\suc n)= (\suc n)* (\mathsf{fac}\,n)\}$, 
we can derive $\vdash_{\Eb_\mathsf{fac}}\ol{\mathsf{fac}}:\fa
x.\N(x)\to\N(\mathsf{fac}\;x)$, where $\ol{\mathsf{fac}}\eqdef \recm s$ and the step term
$s$ is defined as $s\eqdef \lb y\lb z\lb w.\case(w,u.\ol{1},v.\ol{\suc}(yv)\ol{*}(zv))$.
\ei
\end{ejem}
\vspace{5pt}

The reader should convince herself that the naive definition of natural
numbers coming from fixed point theory, given by the predicate transformer
$\Phi\eqdef\lb X.\lb x.x=0\lor X(x)$, does not work. In the equi-inductive approach we
cannot construct any number other than zero, and in the iso-inductive case we
cannot construct the zero. Another possibility is the one taken in
\cite{uus98}, discussed next.
 \begin{ejem}[Equi-inductive Natural Numbers] 
    We define $\N\eqdef \mu(\Phi)$ with the predicate transformer 
   $\Phi=\lb X. \lb x.Z(x)\lor X(p(x))$ where $Z\eqdef\lb x.x=0$ and $p$ is a function
   symbol, whose intended meaning is the predecessor function. We have
 $\vdash\ol{0}:\N(0)$ and $\vdash \ol{p}:\fa x.\N (p(x))\to\N(x)$ where
 $\ol{0}\eqdef \inm(\inl ())\; $ and $\;\ol{p}\eqdef \lb x.\inm(\inr x)$. In this case
 we have the following derivation: $f_0:\fa x.Z(x)\to It(x),\;f_p:\fa x. It(p(x))\to It(x)\vdash \itm s:\fa x.\N(x)\to It(x)$,
 where $s\eqdef\lb x.\lb y.\case(y,u.f_0(u),v.f_p(xv))$ and $It(x)$ is a predicate representing the fact that the image of a given
 function $f$ on $x$ was defined by
 iteration,. If we set $g\eqdef \itm s$ then 
$ g\,\ol{0}\to^\star f_0\,()$ and 
$g\,(\ol{p}\,n)\to^\star f_p\,(g\,n) $. This example shows that our
logic subsumes the Mendler-style programming methodology of \cite{uus98}. However, this approach does not correspond
to the idea of programming with proofs that we pursuit.
 \end{ejem}

Our final version of natural numbers shows the full use of the iso-inductive
style and depends on the disjoint union of predicates $\uplus$ which
is a predicate that can be defined
under the presence of the Parigot's restriction operator
$\restriction$ (see \cite{par92}). This operator can be added to our
logic without a problem and behaves as 
a conjunction where the right formula is an equation without algorithmic
content.\footnote{That is, an equation that is not codified by a proof-term.} Defining
$\Pe\uplus\R\eqdef \lb x.\ex z.(\Pe(z)\!\restriction\! x=\lf
z)\;\lor\;(\R(z)\!\restriction\! x=\rg z)$ we get that $\G\vdash r:\Pe(t)$
implies $\G\vdash\pack(\inl r):(\Pe\uplus\R)(\lf t)$ or $\G\vdash \pack(\inr
r):(\R\uplus\Pe)(\rg t)$. One important advantage of using this predicate
together with our iso-style is that we do not need to deal
directly with existential formulas in definitions, and therefore the following
examples are closer to the data type definition mechanisms of functional programming languages.

 \begin{ejem}[Iso-inductive Natural Numbers]\label{ej:catnats}
  The natural numbers are given now by the inductive definition $\N\eqdef
  \mu(\Phi)$ where $\Phi=\lb X.\unp\uplus X$, and we use a generic
  constructor ${\sf cnat}$, which yields the usual constructors by adopting the equational axioms
  $0={\sf cnat}(\lf\star)$ and $\,\suc x = {\sf cnat}({\sf rg}\,x)$. These
  constructors are implemented by $\ol{0}\eqdef \inm(\pack(\inl ()))$ and $\overline{\suc}\eqdef
 \lb z.\inm(\pack(\inr z))$. Let us present the extracted programs for sum, factorial and predecessor: 
\bi
\item Sum: from $\Eb_{\suma}=\{\suma n\;0 = n,\;\suma n\;(\suc m)=\suc (\suma n\;m)\}$, 
we derive 
$\vdash_{\Eb_{\suma}} \lb n.\itm s:\fa n.\fa x.\N(n)\to\N(x)\to \N(\suma n\;x)$
where
$s\eqdef\lb y.\lb z.\open(z,u.\case(u,v.n,w.\overline{\suc}(yw)))$. Therefore
we get $\ol{\suma}\eqdef\lb n.\itm s$.
\item Factorial: from
$\Eb_{\fac}=\{\fac(0) = 1,\;\fac(\suc(n))= \suc(n)* \fac(n)\}$,
we derive $\vdash_{\Eb_{\fac}} \recm s:\fa x .\N (x)\to \N(\fac x)\;$
where  $s\eqdef \lb y. \lb z.\lb w. \open(w, u.\case(u,
u_1.\ol{1},u_2.\ol{\suc}(y u_2)\ol{\star}(z u_2)))$. Therefore 
$\ol{\fac} \eqdef \recm s$ is a correct program for the factorial.
\item Predecessor: an efficient handling-error predecessor specified by $\Eb_{\pred}
= \{ \mathsf{error}=\lf\,\star,\;\pred\,0 = \mathsf{error},\; \pred\,(\suc\,\allowbreak n) = \rg n\}$, 
is implemented by $\ol{\pred} \eqdef \recm s$, where the step function is
$s \eqdef \lb y.\lb z.\lb w. \open(w,\allowbreak  u.\case(u,u_1.\pack (\inl ())
,u_2.\pack (\inr (yu_2)))$, for we derive
$\vdash\Eb_{\pred} \recm s :$ $ \allowbreak\fa x. \N (x) \to (\unp\uplus \N)(\pred x).$
\ei
 \end{ejem}
\vspace{5pt}

In a similar way to the last example, we can define all usual
inductive data types like finite lists or trees (see \cite{mg11,mp09} for
several related examples). We present next, coinductive predicates
corresponding to the conatural numbers and the lazy data type of streams or
strictly infinite lists. These examples show that we can deal with infinite
objects within a terminating system. It is important to observe that in 
the former case the iso-style is more convenient, and for the latter the equi-style suffices.

The implementation of the predicate for the so-called conatural numbers,
corresponding to the ordinal $\omega+1$, gives us the opportunity to show the use of corecursion to construct
inhabitants of data types with infinite objects, in this case the
ordinal $\omega$. We observe that the implementation of conatural
numbers, as well as the implementations for natural numbers discussed above, do
not correspond to Church numerals, as it happens in $\afd$. In particular the
normal proof-term coding the fact that ${\sf CoNat}(\omega)$ holds does not
involve an ``infinite''  Church
numeral, which would be a non-terminating term, for $\omega$ is specified as a
conatural number that equals its predecessor and will be constructed by means of corecursion.

\vspace{5pt}

\begin{ejem}[Iso-coinductive conatural numbers]
 The conatural numbers are defined by $\mathsf{CoNat}\eqdef\nu(\Phi)$, where
 $\Phi\eqdef \lb X.\lb x.\unp(x)\lor X(x)$, and taking the predecessor function
 {\sf pred} as destructor with implementation $\overline{pred}\eqdef\out$. Let us construct the conatural numbers by means of
 corecursion.
 \begin{itemize}
 \item Zero: let $0$ be a constant, $z$ be a unary function symbol and $\Eb_z=\{{\sf
     pred}(z(x))=\star,\;0=z(\star)\}$. If we define $\bar{0}\eqdef
   \corecm\,s\,()$, where $s\eqdef\lb x\lb y.\lb u.\inl u$, then $\vdash_{\Eb_z}
   \overline{0}:\mathsf{CoNat}(0)$ and 
$\overline{{\sf pred}}\,\overline{0}\to^\star\inl ()$.
\item Succesor: let $\suc$ be a unary function and $\Eb_{\suc}=\{{\sf
  pred}(\suc x)=x\}$. We have $\vdash \overline{\suc}:\fa x.{\sf
    CoNat}(x)\!\to{\sf CoNat}(\suc(x))$, where
  $\overline{\suc}\eqdef\corecm\,s$ and $s\eqdef \lb x\lb y.\lb
  z.\inr(xz)$. Moreover, the operational semantics yields $\overline{\pred}(\overline{\suc}\;n)\to^\star\inr\;n$.
\item Omega: to define the infinite ordinal $\omega$, we use a unary function
  $\omega^\dag$ and axioms
  $\Eb_{\omega^\dag}=\{\omega=\omega^\dag(\star),\;{\sf
    pred}(\omega^\dag(x))=\omega^\dag(x)\}$. Then we get $\vdash
  \overline{\omega^\dag}:\fa x.\unp(x)\to{\sf CoNat}(\omega^\dag(x))$. By
  defining $\overline{\omega}\eqdef\overline{\omega^\dag}\,()$ we get
  $\vdash\overline{\omega}:{\sf CoNat}(\omega)$. The needed proof-term is
  given by $\overline{\omega^\dag}\eqdef\corecm\,s$, where $s\eqdef \lb x\lb
  y\lb z.\inr (yz)$. 
 \end{itemize}
\end{ejem}

Our last example of a coinductive predicate corresponds to streams or  strictly infinite lists.

\begin{ejem}[Equi-coinductive Streams]
  The streams over a data type $\tA$ are defined as
  $\St_\tA\eqdef\nu(\Phi)$ where $\Phi\eqdef\lb X.\lb
  x.\tA(\head(x))\land X(\tail(x))$, and the destructor $d$ is the identity function.
  The programs for the usual destructors are $\overline{\head}\eqdef \lb
   x.\Lp(\out x)$ and $\overline{\tail}\eqdef \lb x.$ $\Rp(\out x)$, extracted from 
$\vdash\overline{\head}:\fa x.\St_\tA(x)\to \tA(\head\,x)$
and $\vdash\overline{\tail}:\fa x.\St_\tA(x)\to \St_\tA(\tail\,x)$.
We present now some programs involving streams:
\bi
\item The function $\from$, that generates the stream of natural numbers from a given one,
is specified by $\Eb_{\from}=\{\head(\from x)=x,\;\tail(\from x)=\from(\suc\, x)\}$. 
The reader can verify that  $\vdash_{\Eb_{\from}}\overline{\from}:\fa x.\N(x)\to \St_\N(\from x)$
where $\overline{\from}\eqdef \coitm s$ and $s\eqdef \lb y\lb
z.\pt{z,y(\overline{\suc}\,z)}$, and that $\ol{\head}(\ol{\from}\,x)\to^\star
x$ and $\ol{\tail}(\ol{\from}\,x)\to^\star\ol{\from}(\ol{\suc}\, x)$. 
\item The constructor $\cons$ is defined by $\Eb_{\cons}=\{ \head (\cons \,x
  \,y) = x,\,\tail(\cons\, x\, y)\! = y\}$ and requires corecursion to
  be implemented. We get a program $\ol{\cons}$ from the proof $\vdash_{\Eb_{\cons}}
  \overline{\cons} : \fa x\fa y. \tA(x)\to \St_\tA(y) \to \St_\tA(\cons\, x \,y)$ 
where $\overline{\cons}\eqdef \lb x.\corecm s$ and $s\eqdef \lb f_1\lb f_2\lb w.\pt{x, f_1 w}$.
\item The function $\map$ on streams is specified by
  $\Eb_{\map}=\{\head(\map \,f \,\ell)=f(\head\ell),\,\tail(\map\,f\ell)\allowbreak=\map\,f\,(\tail\ell) \}$. 
An extracted program from $\vdash_{\Eb_{\map}}\ol{\map} :(\fa x.A(x)\to B(f(x)))\to \fa
  z.\St_A(z)\to\St_B(\map f z)$ is $\ol{\map}\eqdef\lb f. \coitm s $, 
where $s\eqdef \lb y\lb z.\pt{f\,(\ol{\head}\,x) ,y(\ol{\tail}\,x)}$.
\item A function similar to $\map$ but that requires corecursion in the
  implementation is $\maph$, which applies a given function only
  to the head of a stream. It is defined by $\Eb_{\maph} = \{\head(\maph \,f
  \,\ell) = f(\head\, \ell),\;\tail(\maph\,f\,\ell)=\tail\,\ell \}$.
We get the program $\vdash_{\Eb_{\maph}}\ol{\maph}: (\fa x.A(x) \to A(f(x)))\to
\fa z.\St_A(z) \to\St_A(\maph\,f\,z) $ where $\ol{\maph} \eqdef \lb f.\corecm s$
and the step function $\,s $ is defined by $\lb y.\lb z.\lb w.\pt{f\,(\ol{\head}\,x), \ol{\tail}\,x}$.
\ei
\end{ejem}
\vspace{5pt}

We finish the section with a couple of examples involving binary predicates.

\begin{ejem}[Iso-inductive order in natural numbers]
The following recursive definition of order for natural numbers: 
\beqs
\frac{\N(n)}{0 < \suc\, n} \hspace{3cm} \frac{ n < m}{\suc\, n < \suc\,m}
\eeqs
is implemented by the iso-inductive definition $\mathsf{L} = \mu (\Phi)$ where the predicate transformer is
$\Phi \eqdef \lb X^{(2)}.\allowbreak\lb x,y.\,(x= 0\land\N(y)) \lor \exists z. X(z,y)
\!\restriction\!(x= \suc z)$, and the constructors are the identity and the successor
functions $\vec{c}\eqdef\id,\suc$.
The derivations $\;\vdash \lb n.\inm(\inl \pt{(),\allowbreak n}):\fa n. \N(n)\to \mathsf{L}(0,\suc\, n)$ and
$\;\vdash \lb w.\inm(\inr (\pack w)):\fa n \fa m. \mathsf{L}(n,m) \to \mathsf{L}(\suc\,n,\suc\,m)$
can be easily verified.
\end{ejem}

\begin{ejem}[Equi-coinductive observational equality for streams]
Leibniz equality is not always adequate for reasoning about streams
(see \cite{raf93}), in some
cases it is better to employ the observational equality. 
This equality relation is defined by the equi-coinductive binary predicate
$\E\eqdef\nu(\Phi)$ where $\Phi \eqdef \lb X^{(2)}. \lb x,y. \head x=\head y
\land X(\tail x,\tail y)$. It is immediate to verify that 
$\vdash\lb x. \Lp(\out x) :\fa x\fa y. \E(x,y) \to \head\, x = \head \,y\;$
and $\vdash\lb x. \Rp(\out x) :\fa x.\fa y. \E(x,y) \to \E(\tail\,x,\tail
\,y)$. Moreover, the corecursion rule yields $\vdash e:
\fa x\fa y.\head x=\head y\to \E(\tail\,x,\tail\,y)\to\E(x,y)$, where the
proof term $e$ is given by $e\eqdef\lb x\lb
y.\corecm\,s\,\pt{x,y}$ and $s\eqdef \lb w.\lb u.\allowbreak\lb v.\pt{\Lp
  v,w(\Rp v)}$. These proofs imply that two streams are observationally equal if and only if their heads
are equal and their tails are again observationally equal.
\end{ejem}

\section{Saturated Sets}\label{sec:sat}
We develop here all constructions on a complete lattice of so-called
saturated sets needed to define the semantics of the logic.
It is important to emphasize that in this section a term is exclusively a
$\lb$-term belonging to the set $\Lb=\{t\;|\;t\;\mbox{is a proof-term of}\;\afdmn\}$. 
\begin{defn} A term $t$ is called an $I$-term if it was generated by an
  introduction rule, i.e., $I$-terms
are terms of the following shapes:
$\lb xr,\;\inm r,\;\allowbreak\corecm\,s\,r$.
Analogously $E$-terms are terms generated by an elimination rule, i.e. they are terms of the following shapes:
$rs,\;\allowbreak\out r,\;\recm\,s\,r$ .
\end{defn}
\noindent Observe that any term is either a variable, an $I$-term or an $E$-term.\\

Instead of reasoning with infinite reduction sequences we will work with an
inductive definition of a set $\sn$ including all strongly normalizing
terms. We discuss its definition now.

\begin{defn}Evaluation contexts are defined by the following grammar:
\beqs
\ba{rll}
\El{\bullet} & ::= &
\bullet\;|\;\El{\bullet}s\;|\;
\out\El{\bullet}\;|\;\recm s\,\El{\bullet}
\ea
\eeqs
\end{defn}
\noindent Let us observe that an evaluation context may be considered as an
$E$-term with a unique placeholder $\bullet$. Therefore, evaluation
contexts are sometimes called elimination contexts or multiple eliminations. 
In the following,
we will write $\El{r}$ for the $E$-term obtained by substituting the
placeholder $\bullet$ by the term $r$ in $\El{\bullet}$. That is
$\El{r}\eqdef \El{\bullet}[\bullet:=r]$ where the substitution is
defined as if $\bullet$ were a term variable. A term of the form
$\El{x}$ is called a {\em neutral term}. The notion of {\em weak head
  reduction}, denoted $\to_{whd}$, needed to define the set $\sn$ is defined
as follows:
\beqs
\frac{t\to_\beta t'}{\El{t}\to_{whd}\El{t'}}
\eeqs

The final concept involved in the inductive definition of the set $\sn$ is the
set $\ist(t)$ of immediate subterms of a given term $t$, defined as follows:
$\ist(x)=\vacio,\;\ist(\lb xr)=\ist(\inm r)=\ist(\out r)=\{r\},\;
\ist(rs)=\ist(\recm\,s\,r)=\ist(\corecm\,s\,r)=\{s,r\}$. We will also need the
set $\ist(E[\bullet])$ of immediate subterms of a given evaluation context
which is defined as if $E[\bullet]$ were a term.

\begin{defn}
The set $\sn$ is defined by means of the
following inductive definition:
\beqs
\ba{ccc}
\displaystyle\frac{}{x\in\sn}\;\;(\mbox{{\sc sn-var}}) & \;\;\;\;\;\;\;&
\displaystyle\frac{t\;\text{is an $I$-term}\;\;\;\;\;\;{\sf ist}(t)\inc\sn}{t\in\sn}\;\;(\mbox{{\sc sn-i}})
\\ \\
\displaystyle\frac{\El{x}\in\sn\;\;\;\;\;{\sf ist}(\El{\bullet})\inc\sn}{E'\big[\El{x}]\in\sn}\;\;(\mbox{{\sc sn-e}})  &
\;\;\;\;\;\;\;&
\displaystyle\frac{\El{t'}\in\sn \;\;\;\;\;\El{t}\to_{whd}\El{t'}\;\;\;\;\;\sist(t)\inc\sn}{\El{t}\in\sn}\;\;(\mbox{{\sc sn-w}})
\ea
\eeqs
where for a redex $t$, $\sist(t)$ is the set of problematic subterms of
$t$, which are the terms that might break the strong normalization
of $t$, even knowing that its reduct $t'$ strongly normalizes.  This set is defined as follows:
$\sist((\lb x.r)s)=\{s\},\;\sist(\recm s\;(\inm r))=\sist(\out (\corecm s\;r))=\vacio$.
\end{defn}
It can be proved that the characterization $\sn$ of the set of strongly
normalizing terms is sound, that is: if $\,t\in\sn\,$ then there is no infinite
reduction sequence $t\to t_1\to t_2\to\ldots\,$. \\
Now we can define a concept of saturated set, modelled after the definition of $\sn$.

\begin{defn}[$\sat$-set]
A set of terms $\M$ is saturated if and only if it consists only of terms in
$\sn$, 
it contains all neutral terms of $\sn$, and it is closed under
weak head expansion of $\sn$ terms. This can elegantly be 
defined by the following rules:

\beqs
\frac{t\in\M}{t\in\sn}\;\;(\mbox{{\sc sat-sn}})\hspace{1.5cm}\frac{E[x]\in\sn}{E[x]\in\M}\;\;(\mbox{{\sc sat-n}})
\eeqs
\vspace{5pt}
\beqs
\frac{E[t']\in\M\;\;\;\;\;E[t]\to_{whd} E[t']\;\;\;\;\;\sist{(t)} \inc \sn}{E[t]\in\M}\;\;(\mbox{{\sc sat-w}})
\eeqs
\end{defn}

It is easy to see that $\sat\eqdef\{\M\;|\;\M\;\text{is saturated}\}$ is
closed under intersection. Therefore the triple $\pt{\sat,\inc,\bigcap}$ forms
a complete lattice. The next concept will be fundamental for reasoning with saturated sets.
\begin{defn}Given a set of terms $M$, the set
  $\cl{M}:=\bigcap\{\Nc\in\sat\;|\;M\cap\sn\inc\Nc\}$ is called the saturated
  closure or $\sat$-closure of $M$.
\end{defn}
$\cl{M}$ is the least saturated superset of $M\cap\sn$.
Observe that $M\inc\cl{M}$ if and only if $M\inc\sn$.

\subsection{Saturated sets for the implication}
The following construction is standard, we recall it here for the sake of self-containtment.

\begin{defn}
We define $\M\Imp\Nc = \cl{\{r\in\Lb\;|\;\forall s\in\M.\;rs\in\Nc\}}$, so
that $\Imp: \sat\times\sat\to\sat$ is a binary operation on saturated sets.
\end{defn}

\begin{proposition}[Soundness]\label{pr:satcon}
Let $\M,\Nc\in\sat$.
\be
\item If $\mathsf{S}_x(\M,\Nc) = \{t\;|\;\forall s\in\M.\;t[x:=s]\in\Nc\}$ and
  $t\in \mathsf{S}_x(\M,\Nc)$ then $\lb xt\in \M\Imp\Nc$.
\item If $r\in\M\Imp\Nc$ and $s\in\M$ then $rs\in\Nc$.
\ee
\end{proposition}
\begin{proof}
 Straightforward. See for example \cite{mp09}.
\end{proof}

\subsection{$\sat$ valued functions for coinductive predicates}

The goal of this section is to develop the main technical contribution of our
paper, to construct fixed points of $\sat$-valued functions, which will
be needed later for the semantics of coinductive predicates. For the
case of inductive predicates we point to our extended version
\cite{mg11}. The methodology is based on the
one developed in section 9.4 of \cite{mat98} for inductive types. These
constructions and their soundness properties will play an essential role in
the proof  of the adequacy theorem for $\afdmn$.\\

Let us start by fixing a non-empty set $M$ and by defining for all $n\in\N$,
the set of $\sat$-valued $n$-ary functions $\C_n\eqdef
\{F\;|\;F:M^n\to\C\}$, with $\sat_0=\sat$. The set $\sat_n$ forms a complete lattice 
$\pt{\C_n,\inc,\bigcap}$ with the pointwise inherited definitions $F\inc
G\Iff_{def} \fa \vec{x}\in M^n. F(\vec{x}\,)\inc G(\vec{x}\,)$ and defining for any
$\F\inc\C_n$, the function $\bigcap\F:M^n\to\C$ as $(\bigcap\F)(\vec{x}\,)\eqdef
\bigcap_{F\in\F}F(\vec{x}\,)$. 
Through this section we fix a higher-order function $\Phi:\C_n\to\C_n$, and
tuples of functions $\vec{\df}=\df_1,\ldots,\df_n,\vec{\fsf}=\fsf_1,\ldots,\fsf_n$
with $\df_i,\fsf_i:M\to M$.

Let us begin with the constructions for coinductive predicates. The idea is that given a
coinductive predicate $\nu(\Psi)$, where the interpretation of the predicate
transformer $\Psi$ is the function $\Phi:\sat_n\to\sat_n$, its interpretation will be defined as the greatest
fixed-point $\nu(\Theta^\mq)$ of the lower monotonization of some operator $\Theta:\sat_n\to\sat_n$
associated to the arbitrary function $\Phi$.

\begin{defn} We define $\E_\nu:\C_n\to M^n\to\Pe(\Lb)$ by
  $\E_\nu(F)(\vec{t}\,)\eqdef \{r\in\sn\;|\;\out
  r\in\pg(F)(\vec{\df}(\vec{t}\,))\}$ where $F:M^n\to\C$ and $\vec{t}\in M^n$. 
\end{defn}

\begin{lem}\label{enueqthie}
Let $\thie:\C_n\imp\C_n$ be defined as $\thie(F)(\vec{t}\,) =_{def}
\cl{\E_\nu(F)(\vec{t}\,)}$. Then, for any $F\in\C_n, \;\E_\nu(F)=\thie(F)$
\end{lem}
\begin{proof}
It suffices to show that for any $\vec{t}\in M^n,\;\E_\nu(F)(t)\in\sat$. See \cite{mg11}. 
\end{proof}

The post-fixed points of $\thie$ are characterized as follows:
\begin{lem}\label{carnupos} $F\inc\thie(F)\Iff\fa\vec{t}\in M^n\fa r\in
  F(\vec{t}\,).\;\out r\in\pg(F)(\vec{\df}(\vec{t}\,))$.
\end{lem}
\begin{proof}
Straightforward. 
\end{proof}

We would like to obtain a greatest fixed-point of $\thie$, but  as we do not
assume that $\pg$ is monotone, we cannot prove either that $\thie$ is
monotone. 
Therefore we cannot apply the Knaster-Tarski fixed-point
theorem to $\thie$ to obtain a greatest fixed-point of $\thie$, which is what we need to interpret
coinductive predicates. However, we can proceed by using an adequate version
of its lower monotonization (see definition \ref{def:mtz}),
$\E^\subseteq_\nu:\C_n\to M^n\to\Pe(\Lb)$ defined by 
\beqs
\E^\subseteq_\nu(F)(\vec{t}\,)=\bigcap_{F'\in\sat_n}\{\E_\nu(F')(\vec{t}\,)\;|\;F\inc F'\}
\eeqs
It is easy to see that $\E^\subseteq_\nu$ is monotone. Therefore the
operator $\thec:\sat_n\to\sat_n$ given by $\thec(F)(\vec{t}\,)=_{def}\cl{\E^\subseteq_\nu(F)(\vec{t}\,)}$ is
also monotone and the function $\nu(\pif)\in\sat_n$ defined by $\nu(\pif)\eqdef \gfp(\thec)$ exists due to the completeness
of the lattice $\pt{\C_n,\inc,\bigcap}$. 

\begin{proposition}\label{nupos}
  $\nu(\pg)$ is a post-fixed point of $\thie$.
\end{proposition}

\begin{proof}
By definition, $\nu(\pg)$ is a post-fixed point of $\thec$, that is
$\nu(\pg)\inc\thec(\nu(\pg))$. Moreover, it is straightforward to show that
$\thec(\nu(\pg))\inc\thie(\nu(\pg))$, which yields $\nu(\pg)\inc\thie(\nu(\pg))$. 
\end{proof}

Next, we define an operator $\thii$ useful to prove the soundness of the
inference rule for Mendler corecursion.

\begin{defn} Given $\pif:\C_n\imp\C_n$ and $F\in\C_n$ we define
$\I_\nu:\C_n\to M^n\to\Pe(\Lb)$ as follows: if $\vec{s}\in M^n$ and $\vec{s}\neq
\vec{\fsf}(\vec{t}\,)$ then $\I_\nu(F)(\vec{s}\,)\eqdef\vacio$, and
\beqs
\ba{rl}
\I_\nu(F)(\vec{\fsf}(\vec{t}\,))\eqdef 
\Big\{\corecm s\;r\; \Big| & 
H\in\C_n,\;\;r\in H(\vec{t}), \\
& \hspace{-1cm}\;s\in\bigcap_{G\in\C_n}\Big((F\preceq
G)\Imp(H\preceq_{\vec{\fsf}} G)\Imp
H\preceq_{\vec{\df}\circ\vec{\fsf}}\pg(G)\Big)\;\Big\} 
\ea
\eeqs
\noindent
where for any $F,G\in\C_n$ and $\vec{g}$ a tuple of functions $g_i:M\to M$ we define the
$\sat$-set $F\preceq_{\vec{g}} G$ as follows:
$F\preceq_{\vec{g}} G\eqdef
\bigcap_{\vec{t}\in M^n}F(\vec{t}\,)\Imp G({\vec{g}}(\vec{t}\,))$,
in particular, $F\preceq G\eqdef \bigcap_{\vec{t}\in M^n}F(\vec{t}\,)\Imp G(\vec{t}\,)$.\\
Finally we define the function $\thii:\C_n\imp\C_n$ as
$\thii(F)(\vec{t}\,)=_{def}\cl{\I_\nu(F)(\vec{t}\,)}.$
\end{defn}

\begin{lem}\label{lm:inusn}
For any $F\in\C_n,\; \I_\nu(F)\inc\thii(F)$. 
\end{lem}
\begin{proof}
It suffices to show that for any $\vec{s}\in M^n,\;\I_\nu(F)(\vec{s}\,)\inc\sn$. See \cite{mg11}.
\end{proof}

The pre-fixed points of $\thii$ are characterized as follows:

\begin{lem}\label{carnupre} Let $F\in\C_n$. 
\beqs
\ba{rl}
\thii(F)\inc F\Iff &
\fa\, \vec{t}\in M^n.\fa H\in\C_n.\;\fa r\in H(\vec{t}\,). \\
& \;\fa s\in\bigcap_{G\in\C_n}\Big((F\preceq G)\Imp(H\preceq_{\vec{\fsf}} G)\Imp H\preceq_{\vec{\df}\circ\vec{\fsf}}\pg(G)\Big).\corecm s\,r\in F(\vec{\fsf}(\vec{t}\,))
\ea
\eeqs
\end{lem}
\begin{proof} 
Straightforward. 
\end{proof}

To show the soundness of Mendler corecursion we will use the following

\begin{proposition}\label{nupre} 
$\nu(\pif)$ is a pre-fixed point of $\thii$. 
\end{proposition}
\begin{proof} 
We will proceed by extended conventional coinduction, as defined in proposition
\ref{cor:prind}$\,$. \\Let $\Jk\eqdef \nu(\pif)$ and $\Jk'\eqdef \Jk\cup
\thii (\Jk)$. We have to prove that $\thii(\Jk)\inc \thec (\Jk')$ and 
for this, it suffices to show that $\I_\nu(\Jk)(\vec{s}\,)\inc
\E^\subseteq_\nu(\Jk')(\vec{s}\,)$ for all $\vec{s}\in M^n$. \\
If $\vec{s}\neq\vec{\fsf}(\vec{t})$
then $\I_\nu(\Jk)(\vec{s}\,)= \vacio\inc \E^\subseteq_\nu (\Jk')(\vec{s}\,)$. 
For the case $\vec{s}=\vec{\fsf}(\vec{t}\,)$ 
let us take $\corecm s\;r \in \I_\nu (\Jk)(\vec{\fsf}(\vec{t}))$ with $r\in H(\vec{t}),\; H\in \sat_n$ and 
$s\in \bigcap_{G\in\C_n}\big((\Jk\preceq G)\Imp(H\preceq_{\vec{\fsf}} G)\Imp H\preceq_{\vec{\df}\circ\vec{\fsf}}\pg(G)\big)$.
According to the definition of $\E^\subseteq_\nu (\Jk')(\vec{\fsf}(\vec{t}\,))$ we
have to prove that $\corecm s\;r \in \E_\nu (\Jk'')(\vec{\fsf}(\vec{t}\,))$ for any
$\Jk''\in \sat_n$ such that $\Jk'\inc\Jk''$.
Let us observe that $\corecm s\;r \in \sn$, for $\I_\nu (\Jk)(\vec{\fsf}(\vec{t}\,))\inc \sn$. 
Therefore, we only need to verify that $\out (\corecm s\,r)\in\pg (\Jk')(\vec{\df} (\vec{\fsf}(\vec{t}\,)))$.
Since $\pg (\Jk')(\vec{\df} (\vec{\fsf}(\vec{t}\,)))\in\sat$, it suffices to show that
$s(\lb x x)(\corecm s)r\in \pg (\Jk')(\vec{\df}(\vec{\fsf}(\vec{t}\,)))$.\\
We know that $s\in (\Jk\preceq \Jk')\Imp (H\preceq_{\vec{\fsf}} \Jk') \Imp (H \preceq_{\vec{\df} \circ \vec{\fsf}} \pg(\Jk'))$ 
and also that   $\lb x x\in \Jk \preceq \Jk'$, for $\Jk\inc
\Jk'$. Hence, by part 2 of proposition \ref{pr:satcon}$\;$, 
$s(\lb xx)\in H\preceq_{\vec{\fsf}} \Jk' \Imp H \preceq_{\vec{\df} \circ \vec{\fsf}} \pg(\Jk')$.\\
Next, we show that $\corecm s\in H\preceq_{\vec{\fsf}}\Jk'$. By part 1 of
proposition \ref{pr:satcon} we only need to show that for all
$\vec{t}\in M^n,\;\corecm s\;x \in \mathsf{S}_x(H(\vec{t}\,),\Jk'(\vec{\fsf}(\vec{t}\,)))$, which happens
if and only if for all $e\in H(\vec{t}\,),\; (\corecm s\,x)[x:=e] \in \Jk'(\vec{\fsf}(\vec{t}\,))$. 
Therefore we assume $e\in H(\vec{t}\,)$ and need to prove that  $\corecm
s\,e\in\Jk'(\vec{\fsf}(\vec{t}\,))$, but we have $\corecm s\,e\in \I_\nu(\Jk)(\vec{\fsf}(\vec{t}\,))$ and
therefore, by lemma \ref{lm:inusn}$\;$, $\corecm s\,e\in\thii (\Jk)(\vec{\fsf}(\vec{t}\,))$, but
as $\thii(\Jk)(\vec{\fsf}(\vec{t}\,))\inc\Jk'(\vec{\fsf}(\vec{t}\,))$ we have proven that $\corecm s\in H\preceq_\fsf\Jk'$. 
Using again the second part of proposition \ref{pr:satcon}$\;$, we conclude that $s(\lb
xx)(\corecm s)\in H\preceq_{\vec{\df}\circ\vec{\fsf}}\Phi(\Jk')$. Finally $r\in H(\vec{t}\,)$
implies that $s(\lb xx)(\corecm s)r\in\Phi(\Jk')(\vec{\df}(\vec{\fsf}(\vec{t}\,)))$. 
\end{proof}

To finish this section we summarize the soundness properties of the
(co)in\-duc\-ti\-ve constructions on $\sat$-valued functions.

\begin{proposition}[Soundness of the (co)inductive constructions]\label{propsatsetscoind}
 Let $\pif:\C_n\imp\C_n$,
 $\vec{\cf},\vec{\df},\vec{\fsf}$ be tuples of functions
 $\cf_i,\df_i,\fsf_i:M\to M,\;1\leq i\leq n$, and $\vec{t}\in M^n$. Then
\be
\item If $r\in\pif(\mu(\pif))(\vec{t}\,)$ then $\inm r\in\mu(\pif)(\vec{\cf}(\vec{t}\,))$.
\item If $r\in\mu(\pif)(\vec{t}\,),H\in\C_n$ and $s\in
  \bigcap_{G\in\C_n}\Big((G\preceq\mu(\pif))\Imp(G\preceq_{\vec{\fsf}} H)\Imp\pg(G)\preceq_{\vec{\fsf}\circ\vec{\cf}} H\Big)$ then
  $\recm s\,r\in H(\vec{\fsf}(\vec{t}\,))$. 
\item If $r\in\nu(\pif)(\vec{t}\,)$ then $\out r\in\pif(\nu(\pif))(\vec{\df}(\vec{t}\,))$.
\item If $r\in H(\vec{t}\,),H\in\C_n,$ and $s\in
  \bigcap_{G\in\C_n}\Big((\nu(\pif)\preceq G)\Imp(H\preceq_{\vec{\fsf}} G)\Imp H\preceq_{\vec{\df}\circ\vec{\fsf}}\pg(G)\Big)$ then\\
  $\corecm s\allowbreak\,r\in\nu(\pif)(\vec{\fsf}(\vec{t}\,))$. 
\ee
\end{proposition}
\begin{proof}
Part 3 is consequence of proposition \ref{nupos} and lemma \ref{carnupos}$\;$. For
part 4 we just use proposition \ref{nupre} and lemma \ref{carnupre}$\;$. For the
inductive cases we refer to \cite{mg11}$\;$.
\end{proof}

We are now ready to define an intuitionistic semantics for our logic.

\section{Semantics for $\afdmn$}\label{sec:sem}

We present here a realizability semantics for $\afdmn$ where an object-term will be
interpreted as an element of a universe set $M$, a formula as a $\sat$-set and
a predicate as a $\sat$-valued function in $\sat_n$.

\begin{defn}
  A model for a second-order language $\mathfrak{L}$ is a pair
  $\Mf=\pt{M,\I}$  where $M$ is a non-empty set
  and $\I$ is an interpretation function for
  $\mathfrak{L}$ such that $\I(f):M^n\to M$, for every $n$-ary function
  symbol $f\in\mathfrak{L}$ and $\I(P):M^n\to\C$, for every $n$-ary predicate symbol $P\in\mathfrak{L}.$
\end{defn}

From now on we fix a model $\Mf=\pt{M,\I}$.

\begin{defn}
A state or variable assignment is a function $\sigma:Var\to
M\cup\C_n$ such that $\sigma(x)\in M$ and $\sigma(X^{(n)})\in\C_n$.
Given $m\in M$ or $G\in\C_n$, the modified assignments
$\sigma[x/m]$ and $\sigma[X/G]$ are defined as usual.
\end{defn}

Next, we recursively define the interpretation of terms, predicates and formulas.
\begin{defn}
Given a variable assignment $\sigma$, we define the interpretation
function $\ism$, such that $\ism(r)\in M,\;\ism(\Pe)\in\sat_n$ and $\ism(A)\in\sat$, as follows:
\bi
\item Term interpretation 
\bi
\item $\I_\sigma(x)=\sigma(x)$
\item $\I_\sigma(f(t_1,\ldots, t_n))=\I(f)(\I_\sigma(t_1),\ldots,\I_\sigma(t_n))$
\ei
\item Predicate interpretation:
\bi
\item Predicate variables: $\I_\sigma(X)=\sigma(X)$
\item Predicate symbols: $\I_\sigma(P)=\I(P)$
\item Comprehension predicates: if $\F\eqdef\lb\vec{x}A$, we define 
$\I_\sigma(\F)=G_\F$ where $G_\F:M^n\to\sat$ is given by
$G_\F(\vec{m}\,)=\I_{\sigma[\vec{x}/\vec{m}\,]}(A)$, for all $\vec{m}\in M^n$.
\item Predicate transformers: if $\Phi\eqdef\lb X.\Pe$, where w.l.o.g., $\Pe\eqdef\lb\vec{x}.A$, we define
$\I_\sigma(\Phi):\sat_n\to\sat_n$ by
$\I_\sigma(\Phi)(F)(\vec{m}\,)=\I_{\sigma[X/F,\;\vec{x}/\vec{m}\,]}(A)$, for all
$\vec{m}\in M^n$. \\ This way,
it can be proved that for any predicate $\R$, we have $\I_\sigma(\Phi(\R))=\I_\sigma(\Phi)(\I_\sigma(\R))$.
\item (Co)inductive predicates: 
\bi
\item $\I_\sigma(\mu (\Phi))=\mu(\I_\sigma(\Phi))$
\item $\I_\sigma(\nu (\Phi))=\nu(\I_\sigma(\Phi))$
\ei
where of course, the operators $\mu$ and $\nu$ on the right-hand side of the equalities
refer to the constructions on $\sat$-valued functions developed in section \ref{sec:sat}$\;$.
\ei
\item Formula interpretation:
\bi
\item $\I_\sigma\big(\Pe(t_1,\ldots,t_n)\big)=\I_\sigma(\Pe)\big(\I_\sigma(t_1),\ldots,\I_\sigma(t_n)\big)$
\item $\I_\sigma(A\to B)=\I_\sigma(A)\Imp\I_\sigma(B)$
\item $\I_\sigma(\fa xA)=\bigcap\{\I_{\sigma[x/m]}(A)\;|\;m\in M\;\}$
\item $\I_\sigma(\fa XA)=\bigcap\{\I_{\sigma[X/G]}(A)\;|\;G\in\C_n\;\}$
\ei
\ei
\end{defn}

\noindent We observe that as equations are a special case of a second-order universal
formula, there is no need to give a specific semantics for them. However we are only
interested in models that satisfy a set of equations in the following sense.

\begin{defn}
Let $\Mf=\pt{M,\I}$ be a model and $\sigma$ be a state. We say
that the interpretation $\I_\sigma$ satisfies the equation $r=s$ if and only if $\I_\sigma(r)=\I_\sigma(s)$. 
Moreover if $\Eb$ is a set of equations, we say that $\I_\sigma$
satisfies $\Eb$ if and only if $\I_\sigma$ satisfies every equation in $\Eb$.
\end{defn}
\vspace{5pt}
Now we can prove the main theorem of this paper.

\begin{theorem}[Adequacy or soundness] 
Let $\Mf=\pt{M,\I}$ be a model such that the interpretation $\I_\sigma$ satisfies
the set of equations $\Eb$. If $\G\vdash_\Eb t:A$, with $\G=\{x_1:A_1,\ldots,x_n:A_n\}$ and for
all $1\leq i\leq n,\;r_i\in\I_\sigma(A_i)$ then $ t[\vec{x}:=\vec{r}\,]\in \I_\sigma(A)$.
\end{theorem}
\begin{proof}
Induction on $\G\vdash_\Eb t:A$. We discuss the case for the rule $(\nu I)$,
for the remaining rules see \cite{mg11}.
We need to show that $(\corecm s\,r)[\vec{x}:=\vec{r}]\in
\I_\sigma\big(\nu(\Phi)(\vec{f}(\vec{t}\,))\big)$. That is,
$(\corecm s[\vec{x}:=\vec{r}]\, r[\vec{x}:=\vec{r}])\in\nu(\I_\sigma
(\Phi))(\vec{j})$, where $\vec{j}=\I_\sigma(f_1(t_1)),\ldots,\I_\sigma(f_n(t_n))$. The I.H. yields
$s[\vec{x}:=\vec{r}]\in \I_\sigma\Big(\fa X\big(\nu(\Phi)\inc X\imp\K\inc_{\vec{f}} X\imp 
                    \K\inc_{\vec{d}\circ \vec{f}} \Phi(X) \big)\Big)$. 
From this and by defining 
$\Phi'=\I_\sigma(\Phi),\;\vec{\fsf}=\I_\sigma(f_1),\ldots,\I_\sigma(f_n),\;\vec{\df}=\I_\sigma(d_1),\ldots,\I_\sigma(d_n)$ 
and $H=\I_\sigma(\K)$ it is easy to verify that $s[\vec{x}:=\vec{r}]\in  
\bigcap_{G\in\C_n} \Big((\nu(\Phi')\preceq G) \Imp (H \preceq_{\vec{\fsf}} G) \Imp
(H\preceq_{\vec{\df}\circ \vec{\fsf}}\Phi'(G))\Big)$. Moreover we also have
$r[\vec{x}:=\vec{r}]\in H(\vec{l})$, where $\vec{l}=\I_\sigma(t_1),\ldots,\I_\sigma(t_n)$, by I.H. Therefore we can apply part 4
of proposition \ref{propsatsetscoind} to conclude that $\corecm
s[\vec{x}:=\vec{r}]\,r[\vec{x}:=\vec{r}]\in \nu(\Phi'))(\vec{\fsf}(\vec{l}\,))$, 
which is equivalent to $(\corecm s\,r)[\vec{x}:=\vec{r}]\in\I_\sigma(\nu(\Phi)(\vec{f}(\vec{t}\,)))$.
\end{proof}

\subsection{Strong normalization}
The strong normalization property for the logic $\afdmn$ can be proved by adapting the proof of
$\afd$ which embeds this logic into its propositional fragment, system ${\sf F}$
(see \cite{kri93}). However, our semantics of saturated sets allows for an
easy proof of strong normalization which is a direct consequence of 
the adequacy theorem. Let us start by building a model and an
interpretation that satisfies a given set of equations $\Eb$ as required by
the adequacy theorem.

\begin{defn}
  Given a judgement $\Delta=_{def}\G\vdash_\Eb t:A$ we define a model
  $\Mf_\Delta=\pt{M,\I}$ as follows:
  \begin{itemize}
    \item Let $\approx_\Eb$ be the binary relation on terms given by $r=s\Iff_{def}
      \Eb\rhd r=s$. It is easy to prove that $\approx_\Eb$ is an equivalence
      relation. 
  \item The universe of $\Mf_\Delta$ is the set $M={\sf
        Term}_\L/\approx_\Eb$, of the equivalence
      classes $[t]$ of the relation $\approx_\Eb$.
    \item The interpretation function $\I$ is defined as follows:
      \begin{itemize}
      \item $f^\I: M^n\to M,\;
        f^\I([t_1],\ldots,[t_n])=_{def}[f(t_1,\ldots,t_n)]$
      \item $P^\I:M^n\to\sat,\;
        P^\I([t_1],\ldots,[t_n])=_{def}\cl{\{s\in\Lb\;|\;\G\vdash_\Eb s:P(t_1,\ldots,t_n)\}}$
      \end{itemize}
  \end{itemize}
It is easy to see that the interpretation function is well-defined and
therefore $\Mf_\Delta$ is a model.
\end{defn}

The next lemma shows that in $\Mf_\Delta$ term interpretation is given by
a specific substitution.

\begin{lem}\label{lm:isub}
  Let $\sigma$ be a state and $r\in{\sf Term}_\L$  such that $Var(r)=\vec{x}$. If
  $\sigma(x_i)=[s_i]$ then $\I_\sigma(r)=\big[r[\vec{x}:=\vec{s}\,]\big]$.
\end{lem}

\begin{proof}
  Induction on $r$.
\end{proof}

We can now define an interpretation that satisfies a given set of equations $\Eb$.

\begin{lem}\label{lm:msate}
 For any judgement $\Delta=_{def}\G\vdash_\Eb t:A$ there is a state $\sigma$
 of $\Mf_\Delta$ such that the interpretation $\I_\sigma$ satisfies $\Eb$.
\end{lem}

\begin{proof}
  We define the state $\sigma$ of $\Mf_\Delta$, as follows:
  \begin{itemize}
  \item For any first-order variable $x$, $\sigma(x)=[x]$.
  \item For any second-order variable $X$, $\sigma(X)=G$, where 
\beqs
G([t_1],\ldots,[t_n])=_{def}\cl{\{s\in\Lb\;|\;\G\vdash_\Eb
      s:X(t_1,\ldots,t_n)\}}.
\eeqs
\end{itemize}

It is easy to verify that the state is well-defined. Moreover $\I_\sigma$
satisfies $\Eb$, for if $r=s\in \Eb$ then
$r\approx_\Eb s$ and therefore $[r]=[s]$. But, if $Var(r)=\vec{x}$
and $Var(s)=\vec{y}$, then by definition of $\sigma$ and by lemma
\ref{lm:isub} we have $\I_\sigma(r)=\big[r[\vec{x}:=\vec{x}]\big]=[r]=[s]=\big[s[\vec{y}:=\vec{y}]\big]=\I_\sigma(s)$.
\end{proof}

\vspace{.3cm}

The strong normalization of $\afdmn$ is now easily gained from lemma
\ref{lm:msate} and the adequacy theorem.

\begin{theorem}[Strong normalization of $\afdmn$]
If $\G\vdash_\Eb t:A$ then $t$ is strongly normalizing  
\end{theorem}
\begin{proof}
Assume $\Delta$ is the judgement $\G\vdash_\Eb t:A$, with $\G=\{x_1:A_1,\ldots,x_k:A_k\}$. By
lemma \ref{lm:msate} the set of equations $\Eb$ is satisfied by an interpretation
$\I_\sigma$ in the model $\Mf_\Delta$. Moreover, we have $x_i\in\I_\sigma(A_i)$, for
$\I_\sigma(A_i)$ is a $\sat$-set and every $\sat$-set contains all
variables.
Therefore the adequacy theorem yields
$t=t[\vec{x}:=\vec{x}\,]\in\I_\sigma(A)$. Finally, as
$\I_\sigma(A)\inc\sn$, we get $t\in\sn$ which implies that $t$
strongly normalizes.
\end{proof}

\section{Related Work}\label{sec:rw}

Nowadays, there are several lines of research concerning fixed-point logics in computer
science. In relation to our work we can mention for instance \cite{mt03} which
presents a sequent calculus for positive equi-(co)inductive equational definitions and 
which handles conventional (co)iteration only. In this paper the equality relation is primitive
and corresponds to unification with respect to
$\beta\eta$-reduction. Moreover, the
cut-elimination property holds only after restricting
the coinductive rules. Recently \cite{ba11} develops an extension of the
linear logic MALL and a focused proof system for it where the mechanism of
conventional equi-(co)inductive definitions is similar to ours. In this weak
normalizable logic, which only handles (co)iteration, all predicate operators are assumed to be monotone, proofs of functoriality
are given for positive definitions and the treatment of equality originates from
logic programming. Finally we mention the work of \cite{ab07} which is closer
to ours and presents two strongly normalizing propositional logics (type systems) with Mendler-style
positive equi-(co)inductive types whose semantics of so-called guarded
saturated sets makes heavy use of transfinite ordinal recursion, which
obliges to restrict the (co)iteration rules by means of a kind system
that distinguishes between guarded and unguarded types. On the other
hand this feature allows for a definition of a system of sized types
that encompasses primitive (co)recursion and course of value recursion.

\section{Closing remarks}\label{sec:concl}
We have presented the logic $\afdmn$, an extension of the second order logic
$\afd$ with Mendler-style primitive (co)recursion over least and
greatest fixed points of predicate transformers. To our knowledge,
this is the first such extension that includes Mendler-style (co)inductive
predicates while keeping the strong normalization property. Thus, the programs extracted from the 
termination statements of functions are guaranteed to terminate, independently of the syntactical shape 
of the proof and therefore the particular methodologies to show
termination, like the one in \cite{ms94} are not needed. Based on the
concept of monotonization of an operator we have
developed a realizability semantics of $\sat$-sets and $\sat$-valued
functions  for (co)inductive predicates that does not employ the usual
positivity restriction. This was first achieved in \cite{mat98} for
essentially the propositional inductive fragment of our logic. Furthermore, our adequacy theorem does not
require any ordinal recursion in contrast to the work in \cite{par92,raf93}.
The iso-style of our (co)inductive definitions allows to define data types in a similar way
to the definition mechanisms of functional programming by using 
a generic constructor (destructor), a feature that can be easily enhaced 
to use several specific constructors by means of clausular definitions
(see \cite{mp05}), a mechanism which also allows not to use neither
existential nor restricted formulas. By means of examples, we have shown the suitability
of the logic to extract programs from proofs. However, the concept of formal
data type and other semantical foundations of the program extraction method, like the
issue of equality for coinductive data types, as well as the development of more
sophisticated case studies, are work in progress. 

\section*{Acknowledgements}
 We are thankful to the anonymous referees for the helpful comments regarding
 the contents of this paper, in particular for the gentle hint to include the
 conatural numbers as an example. We also gratefully acknowledge Martha Elena Buschbeck Alvarado for
improving the English manuscript.

\end{document}